\documentclass[doublecol]{epl2}

\usepackage{amssymb}
\usepackage{latexsym, amsmath, amscd,amsthm}
\usepackage{graphicx}
\usepackage[percent]{overpic}
\usepackage{units}
\PassOptionsToPackage{caption=false}{subfig}
\usepackage[lofdepth]{subfig}
\usepackage{booktabs} 

\usepackage{clrscode}

\def\figdir{figs/}
\graphicspath{\figdir}

\newtheorem{theorem}{Theorem}

\newtheorem{proposition}[theorem]{Proposition}
\newtheorem{corollary}[theorem]{Corollary}

\theoremstyle{definition}
\newtheorem{definition}[theorem]{Definition}

\newcommand{\R}{\mathbb{R}}

\newcommand{\cov}{\operatorname{cov}}

\newcommand{\norm}[1]{\left\| #1 \right\|}

\newcommand{\dt}{\,\mathrm{d}t}

\newcommand{\from}{\co\!\!}

\newcommand{\tr}{\operatorname{tr}}
\newcommand{\nul}{\operatorname{null}}

\newcommand{\graphG}{\mathbf{G}}
\newcommand{\edgesE}{\mathbf{e}}
\newcommand{\verticesV}{\mathbf{v}}
\newcommand{\edge}{e}
\newcommand{\vertex}{v}
\newcommand{\head}{\operatorname{head}}
\newcommand{\tail}{\operatorname{tail}}

\newcommand{\im}{\operatorname{col}}
\newcommand{\grad}{\operatorname{grad}}
\renewcommand{\div}{\operatorname{div}}

\newcommand{\distas}{\text{ distributed as }}

\def\co{\colon\thinspace}

\title{Factoring the Laplacian to understand topological polymers}
\author{Jason Cantarella\inst{1} \and Tetsuo Deguchi\inst{2} \and Clayton Shonkwiler\inst{3} \and Erica Uehara\inst{4}}
\shortauthor{J. Cantarella \etal}

\institute{
	\inst{1} University of Georgia - Athens, GA, USA\\
	\inst{2} Ochanomizu University - Tokyo, Japan\\
	\inst{3} Colorado State University - Fort Collins, CO, USA \\
	\inst{4} Kyoto University - Kyoto, Japan
}

\abstract{
A ring polymer is a random walk whose steps obey a single linear condition; their sum vanishes. Factoring the graph Laplacian into the product of the incidence matrix and its transpose allows us to show that for a more complicated network, the steps must lie in a linear subspace determined by the graph topology.  
This provides a useful new perspective on the James--Guth theory of phantom elastic networks. In particular, we formulate phantom networks which are free from the constraints of fixed crosslinks. For a given network the solution of the loop constraints makes the partition function finite-valued in the path integral formulation without applying any external forces or fixing any monomer positions. 
The resulting probability distribution on edge displacements is rotationally invariant, which is practically quite useful for generating unbiased random samples of edge displacements and monomer positions. Furthermore, one can exactly calculate many physical quantities such as correlation functions with respect to this distribution. Finally, this reformulation lends itself well to the case of non-Gaussian distributions. We illustrate this by computing the expected radius of gyration of a ring polymer in a wide variety of models.}

\begin{document}

\maketitle

\section{Introduction}

There has been increasing interest in polymers with topologies more complicated than the standard linear polymer in recent years. Branched, multicyclic, ``tadpole'' or ``lasso,'' and bottlebrush polymers have all been studied. Very recently, the Tezuka lab~\cite{Suzuki:2014fo,Tezuka:2017gh} has started to synthesize polymers with even more complicated topologies such as a $K_{3,3}$ graph. 

Polymers are traditionally modeled by random walks, and for previous topologies, the random walk model was relatively simple. The walk is a sum of steps which are either independent (along a branch) or part of a collection of steps conditioned on the hypothesis that they sum to zero (along an isolated loop). This conditioning introduces a small dependence between steps, but the hypothesis is a single linear constraint which can be handled by elementary methods.

For polymers with multiple loops, the steps are conditioned on a hypothesis which is much more complicated-- the sum of steps around \emph{any} loop in the polymer must vanish. Further, the same edge is likely part of many loops at the same time. Understanding the dependency structure of the edges is considerably more complicated in this case and seemed somewhat daunting. This problem was faced in the classical theory of elasticity~\cite{James:1947hp,Flory:1976ke,Haliloglu:1997iu,Estrada:2010fi}, in which the vertex positions $x_i$ were assumed to have the canonical density proportional to 
\begin{equation}
\sum_k e^{-\frac{1}{2}\sum_{i,j \text{ joined by edge }} \norm{x_i^k - x_j^k}^2} = \sum_k e^{-\frac{1}{2}\langle x^k, L x^k \rangle},
\label{eq:laplacian formulation}
\end{equation}
where $L$ is the graph Laplacian and $x^k$ is the vector of $k$-th coordinates of the vertex positions. Here, vertices and edges correspond to crosslinks and chains, respectively, in polymer networks. Eichinger~\cite{eichingerConfigurationStatisticsGaussian1980} gives an excellent summary of this theory. This formulation solves the loop constraint problem by ignoring it; writing everything in term of vertex positions allows one to avoid any discussion of steps. The main remaining problem is that the graph Laplacian $L$ is not positive-definite, since the function above is invariant under translations. In the James--Guth approach some of the monomer positions are fixed in order to make the partition function finite. However, as Eichinger points out, under these extra constraints the calculation of path integrals is quite difficult. 

The graph Laplacian $L$ can be factored as $L=BB^T$ where $B$ is the incidence matrix of the graph. We first rewrite the terms in the density in~\eqref{eq:laplacian formulation} as $e^{-\frac{1}{2}\langle x^k, L x^k \rangle} = e^{-\frac{1}{2}\langle x^k, BB^T x^k \rangle} = e^{-\frac{1}{2}\langle B^T x^k, B^T x^k\rangle}$, suggesting we should shift our attention from $x^k$ to $B^T x^k$. $B^T$ maps from a vector space of vertex position coordinates to a vector space of edge displacement coordinates-- which are called
chain vectors~\cite{Flory:1976ke} in polymer physics-- by taking differences of vertex positions. Since $B^T$ acts like a discrete gradient, this motivates us to recast the problem of loop constraints as something like a vector calculus problem. 

Finding edges which satisfy the constraints is analogous to determining which vector fields on a complicated domain admit a scalar potential: the task is to construct a set of chain vectors that are compatible with all the loop constraints without assuming even the existence of the crosslink positions. This can be done by quite straightforward linear algebra once one chooses to focus on the edge space instead of the vertex space. This new perspective is the primary contribution of our paper. In particular, we observe that selecting vertex positions with center of mass at the origin according to~\eqref{eq:laplacian formulation} is precisely equivalent to selecting edge displacements from a standard Gaussian on the column space of $B^T$. This gives a regularization of the partition function of a given phantom elastic network.
As an application, we give a simple algorithm for sampling Gaussian random embeddings of arbitrary (connected) multigraphs, including a fast algorithm for generating Gaussian ring polymers. This algorithm should be quite useful in physical applications, since it can be used to provide unbiased estimates of the ensemble average of any physical quantity of polymers. 

While the classical theory summarized in~\cite{eichingerConfigurationStatisticsGaussian1980} and extended in~\cite{eichingerRubberElasticitySolution2015} is intimately bound to the Gaussian hypothesis, we have carefully separated linear algebra and probability theory.
Therefore, we can build on this framework, as more physically realistic constraints (self-avoidance, external fields, bending energies, and so forth) are often quite naturally expressed in terms of the edge space. We give an example by computing the expected radius of gyration of the ring polymer in a general model.


In follow-up papers~\cite{tcrw-theory,cantarellaExactFormulaContraction2025}, we use the framework developed here to compute the exact expected radius of gyration and contraction factor of any Gaussian topological polymer modeled on a subdivision of a multigraph $\graphG$ in terms of spectral data associated to $\graphG$. 

The method in the manuscript is quite useful in numerical applications. Formula (5) gives a fast-sampling method for random configurations of ideal topological polymers with loops, as demonstrated in the density functional theory for cyclic block copolymers~\cite{tomiyoshi_density_2024}; it should practically be of the lowest computational cost and much simpler than previous approaches~\cite{des_cloizeaux_topological_1979,Uehara2016}, although excluded volume is not considered. Even the estimates of the mean-square radius of gyration for ideal topological polymers can be useful for explaining the order among estimates for real topological polymers with different architectures such as ring and tadpoles. In simulations the order among estimates of the mean-square radius of gyration for ideal topological polymers with different architectures is often the same as those of real topological polymers with the corresponding architectures but non-zero excluded volume~\cite{Uehara2016}.  

Theoretical results in the manuscript are partially supported by recent experiments~\cite{mato_programmed_2020}. The experimental data of hydrodynamic radius and intrinsic viscosity versus total molecular weight for spiro-multicyclic polymers are plotted in SI of \cite{mato_programmed_2020}. They are at least roughly in agreement with estimates of the mean-square radius of gyration evaluated by subdivision formula~\cite{tcrw-theory}.

\section{Definitions and Preliminaries}

Let $\graphG$ be a connected multigraph with $\edgesE$ edges $\edge_1, \dots, \edge_\edgesE$ and $\verticesV$ vertices $\vertex_1, \dots, \vertex_\verticesV$. We assume that the graph is directed,\footnote{The direction picked for each edge is arbitrary and won't affect the theory. The directions just need to be consistent throughout any particular set of calculations.} so that each edge has a head vertex $\head(\edge_i)$ and a tail vertex $\tail(\edge_i)$. When $\edge_i$ is a loop edge, $\head(\edge_i) = \tail(\edge_i)$, and since we allow multiple edges between vertices, it is no problem if $\head(\edge_i) = \head(\edge_j)$ and $\tail(\edge_i) = \tail(\edge_j)$.

\begin{definition}\label{def:vertex and edge vectors}
	Let $\graphG$ be a connected, directed multigraph. A \emph{vertex vector} $x \in \R^{d\verticesV}$ for $\graphG$ is formed by stacking $\verticesV$ vectors $x_i \in \R^d$, where $x_i$ is the position of vertex $\vertex_i$. We will let $x^k = (x_1^k, \dots , x_\verticesV^k) \in \R^\verticesV$ be the vector of $k$-th coordinates of all vertex positions.
	
	An \emph{edge vector} $w \in \R^{d\edgesE}$ for $\graphG$ is formed by stacking $\edgesE$ vectors $w_j \in \R^d$, where $w_j$ is the displacement along edge~$\edge_j$. We let $w^k = (w_1^k,\dots , w_\edgesE^k) \in \R^\edgesE$ be the vector of all $k$-th coordinates of the edge displacements. See Fig.~\ref{fig:edge and vertex vectors}.	
\end{definition}
\begin{figure}
	\begin{center}
	\includegraphics[width=1in]{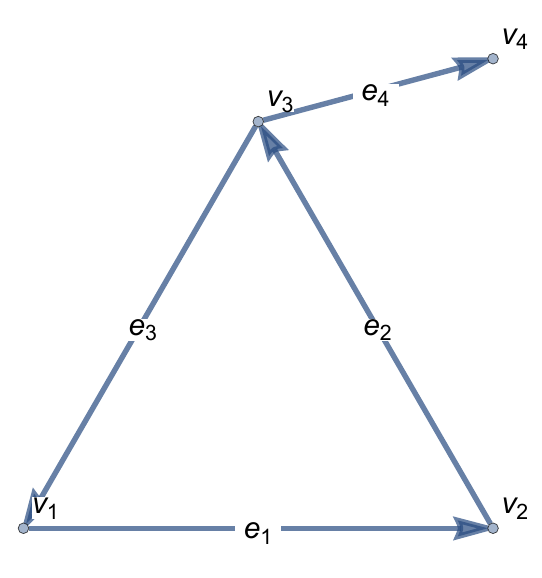} \\
	\begin{tabular}{c}
	  $x = \left( \right.$ \raisebox{-0.35\height}{\includegraphics[width=1.2in]{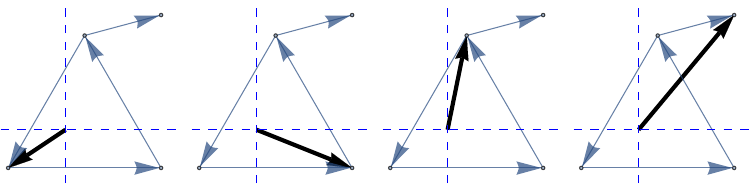}} $\left.\right)$ \\
	  $w = \left( \right.$ \raisebox{-0.35\height}{\includegraphics[width=1.2in]{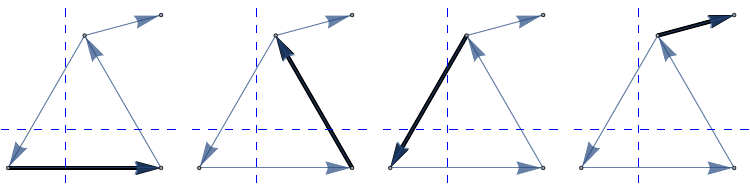}} $\left.\right)$ 
	  \end{tabular}
	  \end{center}
\caption{A particular graph embedding in $\R^2$, along with the components of its vertex vector $x \in (\R^{2})^4$ and edge vector $w \in (\R^{2})^4$.}
\label{fig:edge and vertex vectors} 
\end{figure}

It will also be useful to think of $x^k$ as a (scalar) function $x^k\from \{\vertex_1, \dots, \vertex_\verticesV\} \rightarrow \R$ on the vertices of the graph. We can similarly think of $w^k$ as a function $w^k \from \{ \edge_1, \dots, \edge_\edgesE\} \rightarrow \R$, which we think of as a vector field on the multigraph. Since $\graphG$ is directed, the sign of $w^k(\edge_j)$ uniquely determines a direction of flow along edge $\edge_j$.

We now recall the relationship between scalar functions and vector fields on multigraphs. In analogy to vector calculus, we define two linear maps between the spaces of functions and vector fields. The gradient field $\grad f$ of a function $f \from \{\vertex_1, \dots, \vertex_\verticesV\} \rightarrow \R$ is the vector field defined by 
\begin{equation}
	(\grad f)(\edge_j) = f(\head(\edge_j)) - f(\tail(\edge_j)),
\end{equation}
  and the divergence $\div u$ of a vector field $u\from \{ \edge_1, \dots, \edge_\edgesE\} \rightarrow \R$ is the function
\begin{equation*}
    (\div u)(\vertex_i) = \sum_{j=1}^\edgesE 
    \begin{cases} 
    -u(\edge_j) & \text{if $\vertex_i$ is (only) the head of $\edge_j$,} \\
    +u(\edge_j) & \text{if $\vertex_i$ is (only) the tail of $\edge_j$,} \\
    0 			 & \text{if $\vertex_i$ is head and tail of $\edge_j$,}\\
	0			 & \text{else.}
    \end{cases}
\end{equation*}
As a matrix, $\div = -B$, where $B$ is the $\verticesV \times \edgesE$ incidence matrix with 
\begin{equation}
   B_{ij} = \begin{cases}
   			+1 & \text{if $\vertex_i$ is (only) the head of $\edge_j$,} \\
   			-1 & \text{if $\vertex_i$ is (only) the tail of $\edge_j$,} \\
   			 0 & \text{if $\vertex_i$ is head and tail of $\edge_j$,}\\
   			 0 & \text{else.}
   		 \end{cases}
\end{equation}
Loop edges contribute zero columns, and when there are multiple edges connecting two vertices $B$ will have repeated columns. 

On the other hand, $\grad$ is the $\edgesE \times \verticesV$ matrix $B^T$. We will say that $f$ is a \emph{potential function} for $u$ if $\grad f = u$ and that $u$ is a \emph{conservative} vector field if it has a potential function. If $u$ is conservative, there is a one-dimensional family of possible potential functions $f + C$ where $C$ is a constant function on $\graphG$.

As a subspace of the space of all vector fields, we can characterize the conservative vector fields using a version of the Helmholtz decomposition:
\begin{theorem}
\label{thm:stokes theorem}
The vector space $\R^\edgesE$ of vector fields on $\graphG$ is spanned by a $(\verticesV - 1)$-dimensional subspace of conservative vector fields and an orthogonal $(\edgesE - \verticesV + 1)$-dimensional space of divergence-free vector fields. 
\end{theorem}

\begin{proof}
The divergence-free fields are by definition $\nul B$. Their orthogonal complement is the subspace ${(\nul B)^\perp =  \im B^T = \im \grad}$. Since $\graphG$ is connected, $\nul \grad$ is one-dimensional (only the constant functions have no gradient), so $\im \grad$ has dimension $\verticesV - 1$. The dimension of $\nul \div$ follows.
\end{proof}

Here $\nul B$ (also called $\ker B$) is the subspace of edge space $\R^\edgesE$ whose elements are annihilated by $B$, and $\im B^T$ (also called $\operatorname{im} B^T$) is the subspace of $\R^\edgesE$ generated by the column vectors of $B^T$.

Consider the problem of finding a potential function $f$ given a vector field $u \in \im B^T$. Since $\nul \grad$ is not empty, this problem is underdetermined: adding something in $\nul \grad$ to any particular solution $f$ still yields a function with $\grad f = B^T f = u$. We can define a unique canonical potential function $f$ for $u$ by taking the potential function of minimum norm (among all possible potential functions for $u$).

This minimum norm potential function can be computed conveniently using the Moore--Penrose pseudoinverse. The pseudoinverse $A^+$ of a matrix $A$ is ``as close as possible'' to the inverse of a matrix which is not full rank. It may be computed by taking the singular value decomposition $A = U \Sigma V^T$ and defining $A^+ = V \Sigma^+ U^T$, where $\Sigma^+$ is the diagonal matrix whose nonzero entries are the reciprocals of the corresponding \emph{nonzero} entries in~$S$.

Equivalently, the pseudoinverse is defined to be the matrix $A^+$ satisfying the four \emph{Moore--Penrose conditions}:
\begin{align}
& AA^+A = A, \quad A^+AA^+ = A^+, \label{eq:moore penrose conditions} \\ 
& AA^+ = (AA^+)^T, \quad A^+A = (A^+A)^T. \nonumber
\end{align}
It will also be helpful to recall that $AA^+$ is the orthogonal projector onto $\im A = (\nul A^T)^\perp$, that $A^+A$ is the orthogonal projector onto $\im A^T = (\nul A)^\perp$, and that $(A^T)^+ = (A^+)^T$, so there is no ambiguity in writing $A^{T+}$ for the combination. 

\begin{theorem}
The smallest $x$ minimizing $\|Ax - b\|^2$ is given by $x_0 = A^+ b$. Further, $x_0 \in (\nul A)^\perp$, and $Ax_0 = b$ $\iff$ $b \in \im A$.
\end{theorem}

Using the theorem, we see that if ${u \in \im \grad = \im B^T}$, then there is a unique potential function ${f = B^{T+} u}$ in $(\nul B^T)^\perp$. As noted above, $\nul B^T$ is the one-dimensional\footnote{Recall that $\graphG$ is assumed to be connected.} space of constant functions, so ${f \in (\nul B^T)^\perp}$ must have $\sum f(\vertex_i) = 0$. It is helpful to observe that $B^{T+}u$ is defined for any vector field $u$ (whether or not $u$ is in $\im B^T$), though $B^T B^{T+} w = u$ if and only if $u \in \im B^T$.

The discussion above solves the embedding problem, but seems to avoid mentioning the loops in the graph explicitly. To recover the loops, observe that every loop in the graph has a divergence-free field that flows around it and every vector field perpendicular to that field obeys the corresponding loop constraint. The fields flowing around the loops span $\nul \div$ (for more details, see~\cite{Jiang:2011hk}), whose dimension $\edgesE - \verticesV + 1$ is exactly the cycle rank of the graph.


\section{Gaussian random embeddings}

With the language above, we see that any collection of functions $x^1, \dotsc, x^d \in \R^{\verticesV}$ define an embedding of $\graphG$ into $\R^d$. By contrast, a collection of vector fields $w^1, \dots , w^d \in \R^{\edgesE}$ can always be interpreted as a collection of edge displacement vectors, but these displacements define an embedding of $\graphG$ \emph{only if} $w^k \in \im \grad = \im B^T$ for each $k=1, \dotsc, d$. In this case, the $w$'s define a unique embedding $x^k = B^{T+} w^k$ with $\sum x_i^k = 0$. We call such an embedding \emph{centered} because its center of mass is at the origin.

In analogy to the requirement that the displacements $w$ of a Gaussian random walk are sampled from Gaussians with unit variance on $\R^d$ (that is, they are $\mathcal{N}\left(0,\nicefrac{1}{d}\, I_d\right)$), we make the least restrictive assumption about the distribution of $w^k = \grad x^k = B^T x^k$ that we can: 

\begin{definition}
\label{def:gaussian random embedding}
A \emph{Gaussian random embedding} of $\graphG$ into $\R^d$ with $k$th coordinates of the edge displacements $w^k$ is defined by the assumption that $w^k$ is sampled from $\mathcal{N}\left(0,\nicefrac{1}{d}\, I_\edgesE\right)$ conditioned on the hypothesis that ${w^k \in \im \grad = \im B^T}$.
\end{definition}

We now compute the covariance matrix of $w^k$:
\begin{theorem}
\label{thm:displacement covariances}
If $w$ is the edge vector of a Gaussian random embedding of $\graphG$ into $\R^d$, then the vector $w^k$ of $k$th coordinates is distributed as $\mathcal{N}\left(0,\nicefrac{1}{d}\,B^TB^{T+}\right) = \mathcal{N}\left(0,\nicefrac{1}{d}\,B^+B\right)$. 
\end{theorem}

That is, the $k$th edge displacement coordinates $w_i^k$ follow the Gaussian distribution with covariance matrix $\langle w_i^k, w_j^k \rangle = (B^+B)_{ij}/d$ for $i, j = 1, 2, \dots , \edgesE$.  

\begin{proof}
To condition on the hypothesis that a multivariate normal is restricted to a linear subspace, we transform the normal by orthogonal projection to that subspace. Using the fact that $B^T B^{T+}$ is the (symmetric) orthogonal projector onto $\im B^T$ we can compute that the covariance matrix of the projected variable is $\nicefrac{1}{d}$ times
\begin{equation*}
B^T B^{T+} I (B^T B^{T+})^T = B^T B^{T+}= (B^+ B)^T = B^+B.
\end{equation*}
This completes the proof.
\end{proof}
 
We can show that the probability distribution function of edge vectors is invariant under rotation in im $B^T$. Moreover, it is straightforward to calculate the correlation between any given pair of edges, even if they are located in globally distant and separate regions of a network.    

For an illustration, let us consider a multi-theta graph $\Theta_{m,n}$, which consists of two branch points with functionality $m$ where each of the $m$ branches connecting them has $n$ edges. We assume that each edge vector has the same orientation from one branch point to another one along the branches. Then, computing the orthogonal projector $B^+B$ allows us to show:

\begin{corollary}\label{cor:theta correlation}
	Let $w_i$ and $w_j$ be the displacement vectors of edges $e_i$ and $e_j$ in a Gaussian embedding of $\Theta_{m,n}$. Then
	\begin{equation*}
		\mathcal{E}(w_i^T w_j) = \begin{cases} 1-\frac{m-1}{mn} & \text{ if } i = j, \\
		-\frac{m-1}{mn} & \text{ if } e_i \text{ and } e_j \text{ on same branch,} \\
		\frac{1}{mn} & \text{ otherwise.} 
		\end{cases}
	\end{equation*}
\end{corollary}

It is easy to see that the result generalizes that of linear and ring polymers for $m=1$ and $m=2$, respectively.  

Given the edge displacements, it is certainly possible to determine the vertex vector of the corresponding centered embedding directly: for the path, the vertex positions are just the partial sums of the displacement vectors. Strictly speaking, this is based on a choice of (the unique) spanning tree for the path graph, but for more complicated multigraph topologies one needs to find a spanning tree before computing the partial sums. It is generally simpler to compute $B^{T+}$ and use the equation $x^k = B^{T+} w^k$. 

\begin{corollary}
If $x$ is the vertex vector of a Gaussian random embedding of $\graphG$ into $\R^d$, the vector $x^k$ of $k$-th coordinates can be constructed by taking $x^k = B^{T+} y^k$ where $y^k$ is distributed as $\mathcal{N}\left(0,\nicefrac{1}{d}\,I_\edgesE\right)$ on $\R^\edgesE$.
\label{cor:vertex positions from y edges}
\end{corollary}

\begin{proof}
We can construct $x^k$ by taking ${x^k = B^{T+} w^k}$, where $w^k$ is distributed as $\mathcal{N}\left(0,\nicefrac{1}{d}\,B^T B^{T+}\right)$ on $\R^\edgesE$. In the proof of Theorem~\ref{thm:displacement covariances}, we showed that $w^k = B^T B^{T+} y^k$, where $y^k$ is distributed as $\mathcal{N}\left(0,\nicefrac{1}{d}\,I_\edgesE\right)$. Thus, using~\eqref{eq:moore penrose conditions}, 
\begin{align*}
	x^k & = B^{T+} B^T B^{T+} y^k = B^{T+} y^k.\qedhere
\end{align*}
\end{proof}

\begin{corollary}\label{cor:vertex distance}
If $\vertex_i$ and $\vertex_j$ are vertices of $\graphG$ connected by a path $p = \edge_1 + \dots + \edge_{n}$ of edges and $\ell_1, \dots, \ell_{\chi(\graphG)}$ is an orthonormal basis for the loop space $\nul B$, then the expectation $\mathcal{E}(\norm{x_i - x_j}^2) = n - \sum_{i=1}^{\chi(\graphG)} \left(\ell_i^T p\right)^2$.
\end{corollary}

\begin{proof}
Thinking of $\vertex_1,\dots , \vertex_\verticesV$ as standard basis vectors for $\R^\verticesV$, for each coordinate $k$, we can write $x^k_i - x^k_j$ as $\left(x^k\right)^T (v_i - v_j)$. But then, since $x^k =B^{T+} y^k$, we have $$(x^k_i - x^k_j)^2 = \left(\left(y^k\right)^T B^+Bp \right)^2.$$ Since $y^k$ is distributed as $\mathcal{N}\left(0,\nicefrac{1}{d}\,I_{\edgesE}\right)$, the expectation is equal to $\frac{1}{d}\norm{B^+ B p}^2$. $B^+ B$ is the orthogonal projector onto $(\nul B)^\perp$, so $B^+B p = p - \sum \left(\ell_i^T p \right) \ell_i$. The statement follows from orthonormality of the $\ell_i$ and ${\norm{p}^2 = n}$. 
\end{proof}

Returning to the example of $\Theta_{m,n}$, let $a_i$ be the $i$th arc connecting junction vertices. $\Theta_{m,n}$ has cycle rank $m-1$, and an orthonormal basis $\ell_1, \dots , \ell_{m-1}$ for $\nul B$ is
\begin{equation*}
	\ell_i = \frac{1}{\sqrt{(m-i)(m-i-1)}}\left((m-i)a_i - a_{i+1} - \dots - a_m\right).
\end{equation*}
With this basis, $p=a_1$ has $\ell_i^Tp = 0$ for $i>1$, and Corollary~\ref{cor:vertex distance} implies that the expected squared distance between junction vertices is $\frac{n}{m}$, as expected~\cite{Uehara:2018bb,zhu_radius_2016}.

We now compute the covariance matrix of~$x^k$. 

\begin{theorem}
\label{thm:position covariances}
If $x$ is the vertex vector of a Gaussian random embedding of $\graphG$ into $\R^d$, then the vector $x^k$ of $k$th coordinates is distributed as $\mathcal{N}\left(0,\nicefrac{1}{d}\,L^+\right)$ where $L$ is the graph Laplacian of $\graphG$. Hence $x^k$ may be constructed by
\begin{equation}
x^k = (L^+)^{1/2} y^k
\label{eq:square root construction}
\end{equation}
where $(L^+)^{1/2}$ is any symmetric square root of $L^+$ and $y^k$ is distributed as $\mathcal{N}\left(0,\nicefrac{1}{d}\,I_\verticesV\right)$ on $\R^\verticesV$.
\end{theorem}

That is, random sampling of the $k$th coordinate of vertex vector $x$ (denoted by $x^k$) is obtained by multiplying the square root of $L^+$ by a vertex vector $y^k$ such that each of the $\verticesV$ components is given by the normal distribution of variance $1/d$.

\begin{proof}
For the first part, we know that $x^k = B^{T+} y^k_\edgesE$, where $y^k_\edgesE$ is distributed as $\mathcal{N}\left(0,\nicefrac{1}{d}\,I_\edgesE\right)$ on $\R^\edgesE$. Therefore $x^k$ is a multivariate normal whose covariance matrix is 
\begin{equation*}
B^{T+} \left(\frac{1}{d}I_\edgesE\right) (B^{T+})^T = \frac{1}{d}B^{T+} B^+ = \frac{1}{d}(BB^T)^+ = \frac{1}{d}L^+,
\end{equation*}
using the Moore--Penrose conditions and the fact that $(AA^T)^+ = A^{T+}A^+$ for any $A$. The construction~\eqref{eq:square root construction} is justified by the fact that $(L^+)^{1/2} y^k$ is also a multivariate normal with covariance matrix $\nicefrac{1}{d}\,L^+$. Note that $L$ is a real symmetric matrix, so it has a singular value decomposition in the form $U \Sigma U^T$. This means we can let $(L^+)^{1/2} = U^T (\Sigma^{+})^{1/2} U$. 
\end{proof}

We can sample random vectors $x^k$ of vertex coordinates in two ways: using Corollary~\ref{cor:vertex positions from y edges} or using Theorem~\ref{thm:position covariances}. The latter is almost always preferable, since any multigraph with cycles has at least as many edges as vertices and the covariance matrix $L^+$ only has to be computed once. This, then, gives a powerful computational tool for estimating arbitrary quantities using Monte Carlo integration.

We now do an example. For ring polymers, the graph $\graphG$ is a cycle graph and the Laplacian is:
\begin{equation*}
 L =
  \begin{pmatrix}
       2 &     -1 & \cdots & \cdots &  -1 \\
      -1 &      2 & -1     & \cdots &  0 \\
       0 &     -1 &  2     & -1     &  0 \\
  \vdots & \vdots &  \ddots& \ddots & -1 \\
      -1 & \cdots &  \cdots& -1     &  2
  \end{pmatrix}
 \end{equation*}
The circulant matrix $L$ can be diagonalized as $L = U \Sigma U^T$ where $\sigma_{j} = 4 \sin^2 \frac{\pi j}{\verticesV}$. Since $\sigma_j = \sigma_{\verticesV -j}$, most eigenspaces of $L$ are two-dimensional and the matrix $U$ is not uniquely determined. Eichinger~\cite{Eichinger:1972iy} uses the real matrix
\begin{equation}
U_{ij} = \frac{1}{\sqrt{\verticesV}} \left( \cos \frac{2\pi i j}{\verticesV} + \sin \frac{2 \pi i j}{\verticesV} \right),
\label{eq:svd for cycle graph laplacian}
\end{equation}
which is the discrete Hartley transform; the complex discrete Fourier transform $\widetilde{U}_{ij} = \frac{1}{\sqrt{\verticesV}} e^{-\frac{2\pi \sqrt{-1}}{\verticesV} (i-1)(j-1)}$ also yields $L = \widetilde{U} \Sigma \widetilde{U}^*$.

This allows us to construct a square root $(L^+)^{1/2}$ whose singular values are $\frac{1}{2} \csc \frac{\pi j}{\verticesV}$ for $1 \leq j < \verticesV$ together with a single $0$ corresponding to $\sigma_\verticesV = 0$. We note a similar description for Gaussian ring polymers was also considered by Bloomfield and Zimm~\cite{Bloomfield:1966cy}.

\section{Fourier-type analysis of random graph embeddings}

In the proof of Theorem~\ref{thm:position covariances}, we saw that the vertex vectors of a Gaussian random embedding could be generated by taking
\begin{equation*}
x^k = U^T (\Sigma^+)^{1/2} U y^k, \quad y^k \distas \mathcal{N}\left(0,\nicefrac{1}{d}\,I_\verticesV\right),
\end{equation*}
where $U$ is the matrix of eigenvectors of $L$. Since $U$ is orthogonal, $U y^k$ is also distributed as $\mathcal{N}\left(0,\nicefrac{1}{d}\,I_{\verticesV}\right)$, and we may generate samples of $x^k$ efficiently by multiplying Gaussian random variates by $U^T (\Sigma^+)^{1/2}$. In the case of ring polymers, this can be done in $O(\verticesV \log \verticesV)$ time using the fast Hartley~\cite{Bracewell1984} or fast Fourier transform.

A more geometric way to look at this equation is to see that $x^k$ is a weighted linear combination of the eigenvectors with random (normal) coefficients where the weights are given by the singular values on the diagonal of $(\Sigma^+)^{1/2}$:
\begin{equation*}
x^k = \sum_{\substack{j \in \{ 1,\dotsc,\verticesV \} \\ \sigma_j \neq 0}} \frac{y_j}{\sqrt{\sigma_j}} u_j
\end{equation*}
where each $y_j$ is distributed as $\mathcal{N}\left(0,\nicefrac{1}{d}\right)$. It is clear that the eigenvectors of $L$ with small eigenvalues $\sigma_j$ are expected to play a much larger role in determining the vertex vector $x^k$ than those with larger eigenvalues. 

We can make this observation precise by recalling a few facts from linear algebra.
An optimal rank-$p$ approximation $A_p$ to a matrix $A$ with singular value decomposition $U \Sigma V^T$ is given by replacing $\Sigma$ with another diagonal matrix $\Sigma_p$ keeping a collection\footnote{This collection is not always unique if the singular values of $A$ are not all distinct; in this case, all matrices $A_p$ constructed in this way are equally good rank-$p$ approximations to $A$.} $S_p$ of the $p$ largest singular values of $\Sigma$ and setting the remaining singular values to zero. We can now define the rank-$p$ approximation to a Gaussian random graph embedding
\begin{equation*}
x^k_p = (L^+)^{1/2}_p y^k
\end{equation*}
and note that 
\begin{equation*}
x^k - x^k_p = ((L^+)^{1/2} - (L^+)^{1/2}_p) y^k
\end{equation*}
is a Gaussian random vector. If $\sigma_j^+$ are the singular values of $L^+$, and $S^+_p$ is a collection of the $p$ largest $\sigma_j^+$, then the expected squared norm of this difference vector is 
\begin{equation}
\mathcal{E}\left(\norm{x^k - x^k_p}^2\right) = \frac{1}{d}\tr\left(L^+ - L^+_p\right) =\frac{1}{d} \sum_{\sigma_j^+ \not\in S_p} \sigma_j^+.
\label{eq:p error estimate}
\end{equation}
For the cycle graph, we know from~\eqref{eq:svd for cycle graph laplacian} that $\sigma_j^+=\frac{1}{4} \csc^2 \frac{\pi j}{\verticesV}$ for $j \in \{1, \dots, \verticesV-1\}$ and $\sigma_\verticesV^+=0$. Summing the $\sigma_j^+$ using the formula $\sum_{j=1}^{v-1} \csc^2 (\pi j)=(1/3)(v^2-1)$~\cite[4.4.6.5]{Prudnikov:1986vp} tells us that $\mathcal{E}(\norm{x^k}^2) = \frac{1}{12d}(\verticesV^2 - 1)$ and (for $p$ even) 
\begin{multline}\label{eq:error estimate}
	\mathcal{E}\left(\norm{x^k - x^k_p}^2\right) =\frac{1}{d} \sum_{j=\nicefrac{p}{2}+1}^{\verticesV-\nicefrac{p}{2}-1} \sigma_j^+ \\
	\approx \frac{1}{d}\!\!\!\! \int\limits_{\frac{p+1}{2}}^{n-\frac{p+1}{2}} \frac{1}{4} \csc^2 \frac{\pi t}{\verticesV} \dt = \frac{\verticesV}{2 \pi d} \cot \left( \frac{\pi}{2} \frac{p+1}{\verticesV} \right).
\end{multline} 
Fig.~\ref{fig:ring example} shows examples of low rank approximations to the ring polymer with $\verticesV = 1000$ for $p \in \{20,50,250,999\}$.

\begin{figure}
\hphantom{.}
	\raisebox{-0\height}{\begin{minipage}{1.6in} 
	\includegraphics[height=1.25in]{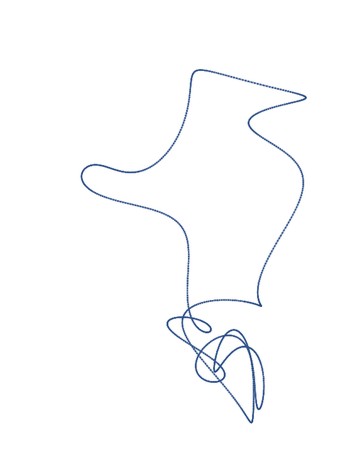} \\ 
	\vspace{-0.1in}
	\begin{center} $x^k_{20} = (L^+)^{1/2}_{20} y^k$ \end{center} \end{minipage}}
\hfill 
	\raisebox{-0\height}
	{\begin{minipage}{1.6in} 
	\includegraphics[height=1.25in]{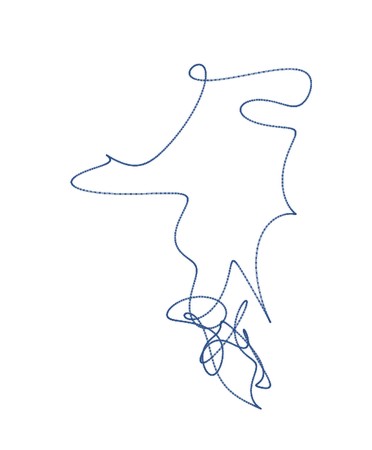} \\ 
	\vspace{-0.1in}
	\begin{center} $x^k_{50} = (L^+)^{1/2}_{50} y^k$ \end{center} \end{minipage}}
\\ 
	\raisebox{-0\height}{\begin{minipage}{1.6in} 
	\includegraphics[height=1.25in]{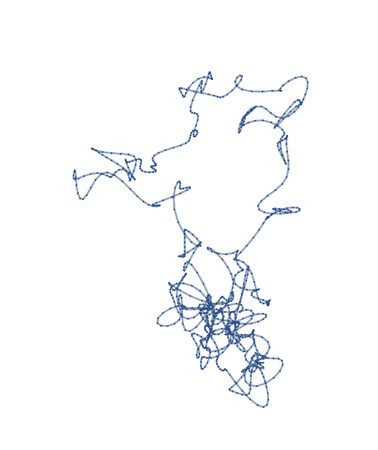} \\ 
	\vspace{-0.1in}
	\begin{center} $x^k_{250} = (L^+)^{1/2}_{250} y^k$ \end{center} \end{minipage}}
\hfill
	\raisebox{-0\height}{\begin{minipage}{1.6in} 
	\includegraphics[height=1.25in]{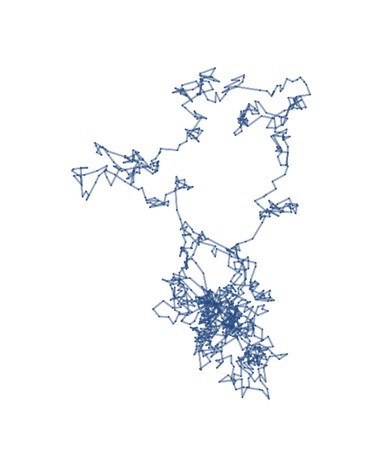} \\ 
	\vspace{-0.1in}
	\begin{center} $x^k = (L^+)^{1/2} y^k$ \end{center} \end{minipage}}
\hphantom{.}
	
\caption{In these pictures, we construct a Gaussian ring polymer $x \in \R^3$ with $1000$ edges by sampling $y^1, y^2, y^3 \in \R^{1000}$ from $\mathcal{N}\left(0,\frac{1}{3}I_{1000}\right)$ and computing $x^k_p = (L^+)^{1/2}_p y^k$ for various low-rank approximations $(L^+)^{1/2}_p$ of $(L^+)^{1/2}$. We can see that the low-rank approximations model the polymer rather well.
}
\label{fig:ring example} 
\end{figure}
 
We see a similar phenomenon for a $\theta$-curve where each edge has been subdivided $500$ times. Here, we can compute the spectrum of the Laplacian numerically and calculate expected error estimates using~\eqref{eq:p error estimate}, but we do not have explicit formulae for these estimates as we did for the cycle. The results are shown in Fig.~\ref{fig:theta example}.

\begin{figure}
\hphantom{.}
	\raisebox{-0\height}{\begin{minipage}{1.63in} 
	\includegraphics[height=1.25in]{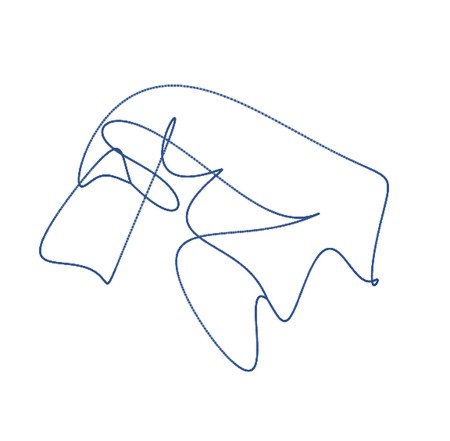} \\ 
	\vspace{-0.1in}
	\begin{center} $x^k_{30} = (L^+)^{1/2}_{30} y^k$ \end{center} \end{minipage}}
		\hfill
	\raisebox{-0\height}
	{\begin{minipage}{1.63in} 
	\includegraphics[height=1.25in]{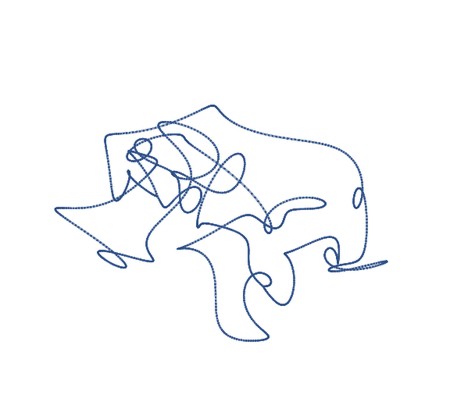} \\ 
	\vspace{-0.1in}
	\begin{center} $x^k_{75} = (L^+)^{1/2}_{75} y^k$ \end{center} \end{minipage}}
\\ 
	\raisebox{-0\height}{\begin{minipage}{1.63in} 
	\includegraphics[height=1.25in]{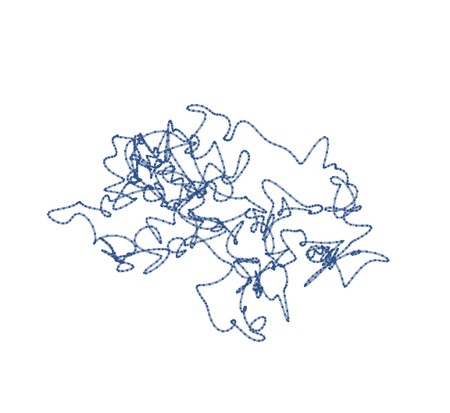} \\ 
	\vspace{-0.1in}
	\begin{center} $x^k_{375} = (L^+)^{1/2}_{375} y^k$ \end{center} \end{minipage}}
		\hfill
	\raisebox{-0\height}{\begin{minipage}{1.63in} 
	\includegraphics[height=1.25in]{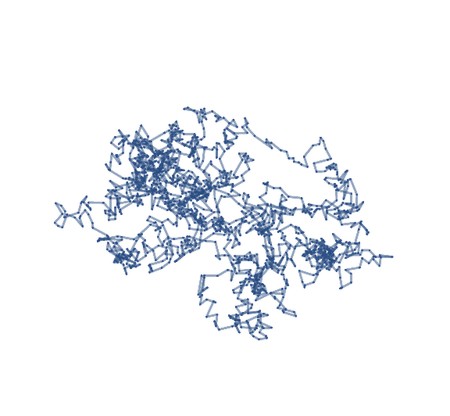} \\ 
	\vspace{-0.1in}
	\begin{center} $x^k = (L^+)^{1/2} y^k$ \end{center} \end{minipage}}
\hphantom{.}
	
\caption{In these pictures, we construct a Gaussian random embedding $x \in \R^3$ of a 1499 vertex $\theta$-curve where each edge has been subdivided into 500 pieces by sampling $y^1, y^2, y^3 \in \R^{1499}$ from $\mathcal{N}\left(0,\frac{1}{3}I_{1499}\right)$ and computing $x^k_p = (L^+)^{1/2}_p y^k$ for various low-rank approximations $(L^+)^{1/2}_p$ of $(L^+)^{1/2}$. Again, the low-rank approximations model the polymer rather well. 
}
\label{fig:theta example} 
\end{figure}

\section{Non-Gaussian Random Graph Embeddings}

This conceptual framework extends to handle much more general distributions, though doing so requires us to stop treating different coordinates as independent: while a multivariate Gaussian is simply a product of scalar Gaussian distributions on the coordinates, this is not true of general multivariate distributions.

It follows from the properties of the Kronecker product $\otimes$ that $x = (B^{T+} \otimes I_d) w$ because each $x^k = B^{T+} w^k$. Hence, the preceding discussion is a coordinate-wise analysis of the conditional distribution of $\mathcal{N}(0,\nicefrac{1}{d}\, I_{d\edgesE})$ given that $w \in \im(B^T \otimes I_d)$. Stated like this, it is clear how to generalize: given any distribution on $\R^{d \edgesE}$, we get a distribution $\mu$ on embeddings of $\graphG$ in $\R^d$ by conditioning on the hypothesis that the edge vector $w$ is in $\im (B^T \otimes I_d)$. Conditioning on this linear constraint is often relatively simple. We may then construct the corresponding probability distribution on vertex vectors $x \in \R^{d\verticesV}$ by pushing $\mu$ forward by the linear transformation $B^{T+} \otimes I_d$.


\begin{proposition}\label{prop:invariant convariance}
If $\mu$ is invariant under the action of $O(d)$ on each edge displacement vector $w_i \in \R^d$, then the $d \edgesE \times d \edgesE$ covariance matrix of $w$ is given by 
$$
\cov(w) = \cov(w^k) \otimes I_d,
$$
where $\cov(w^k)$ is the $\edgesE \times \edgesE$ covariance matrix of $w^k$; that is, $\cov(w^k)_{ij} = \langle w_i^k w_j^k\rangle$. In particular, all $\cov(w^k)$ are equal.
\end{proposition}

\begin{proof}
If $A \in O(d)$, $A$ acts on $w$ by $w \mapsto (I_{\edgesE} \otimes A)(w)$. Since the distribution of $w$ is invariant under this linear map, the covariance matrix $\cov(w)$ is invariant under conjugation by this map:
\begin{equation*}
\begin{aligned}
\cov(w) &= \cov((I_{\edgesE} \otimes A)w) = (I_{\edgesE} \otimes A) \cov(w) (I_{\edgesE} \otimes A)^T 
\end{aligned}
\end{equation*} 
We may write the $d \edgesE \times d \edgesE$ matrix $\cov(w)$ as a $\edgesE \times \edgesE$ block matrix of $d \times d$ matrices. Conjugation by $I_{\edgesE} \otimes A$ conjugates each $d \times d$ submatrix by $A$. Since each of these matrices is fixed by this conjugation, it must be a scalar multiple of $I_d$. In particular, selecting the $k$-th diagonal entry from each block yields $\cov(w^k)$ (and hence these are all equal).
\end{proof}

\begin{proposition} 
With $\mu$ as in Proposition~\ref{prop:invariant convariance}, the expected squared radius of gyration is given by 
$$
\begin{aligned}
\langle R^2_g \rangle &= \frac{1}{\verticesV} \tr \cov(x) \\
&= \frac{d}{\verticesV} \tr \cov(x^k) = \frac{d}{\verticesV} \tr(B^{T+} \cov(w^k) B^{+}).
\end{aligned}
$$
\end{proposition}

\begin{proof}
Since the vertex vector $x$ has center of mass at the origin, ${\langle R^2_g \rangle = \frac{1}{\verticesV} \mathcal{E}(\norm{x_i}^2) = \frac{1}{\verticesV} \sum_{i,k} \mathcal{E}((x^k_i)^2)}$, proving the first equality.
In turn, $x = (B^{T+} \otimes I_d) w$, so Proposition~\ref{prop:invariant convariance} proves the second two equalities. 
\end{proof}

These propositions are true for virtually any model of interest in polymer science. For instance, self-avoiding polymers must have covariance matrices and expected radii of gyration in this form. This allows us to prove something quite general about models of the ring polymer where the edges are drawn from the same distribution (subject to the $\im(B^T \otimes I_d)$ constraint). 

\begin{theorem}\label{thm:R_g ring}
Suppose that $\graphG$ is a cycle graph, and $\mu$ is a probability distribution on $\im(B^T \otimes I_d)$ which is invariant under the action of $O(d)$ and the action of the permutation group on edges in the graph. The variance $\lambda^2$ of each edgelength is the same, and
$$
\langle R^2_g \rangle = \frac{\lambda^2}{12} (\edgesE + 1).
$$
\end{theorem}

\begin{proof}
Since $\mu$ is invariant under the action of the permutation group, the covariance matrix $\cov(w) = \cov(w^k) \otimes I_d$ is unchanged when conjugated by any $P \otimes I_d$ where $P$ is a permutation matrix. This action conjugates $\cov(w^k)$ by $P$. Since $\cov(w^k)$ is fixed by this for any $P$, it follows that all off-diagonal entries are equal, as are all diagonal entries. Thus $\cov(w^k) = \alpha I_{\edgesE} + \beta 1_{\edgesE \times \edgesE}$ for some $\alpha$ and $\beta$.

Now $\nul B$ is spanned by the vector $1_{\edgesE}$ of $1$'s. Since $w^k \in (\nul B)^\perp$, $0 =\mathcal{E}\left(\left( 1_{\edgesE}^T w^k \right)^2\right) = 1_{\edgesE}^T \cov(w^k) 1_{\edgesE} $. Thus $1_{\edgesE} \in \nul \cov(w^k)$. In particular, all of the row and column sums of $\cov(w^k)$ are zero. Coupled with the above, we see
$$
\cov(w^k) = \frac{\lambda^2}{d} \frac{\edgesE}{\edgesE - 1} I_\edgesE - \frac{\lambda^2}{d} \frac{1}{\edgesE-1} 1_{\edgesE \times \edgesE}.
$$
$B^T B$ is a symmetric circulant matrix; its first row is $2, -1, \dots, 0, -1$. The row sums of $B^T B$ vanish, so the column sums of $(B^T B)^+$ also vanish and $1_{\edgesE \times \edgesE} (B^T B)^+ = 0$. 

Now $\tr (B^T B)^+ = \frac{1}{12} (\edgesE^2-1)$ (again using~\cite[4.4.6.5]{Prudnikov:1986vp}). Finally, we compute
\begin{align*}
\langle R^2_g \rangle &= \frac{d}{\edgesE} \tr \left(\cov(w^k) (B^T B)^+ \right) \\
 	      &= \frac{d}{\edgesE} \tr \left(\frac{\lambda^2}{d} \frac{\edgesE}{\edgesE-1} (B^T B)^+ \right) = \frac{\lambda^2 (\edgesE + 1)}{12}. \qedhere
\end{align*}

\end{proof} 

The Gaussian case obeys all the hypotheses. Using Corollary~\ref{cor:vertex distance}, the loop space $\nul B$ is $1$-dimensional and generated by $\ell_1 = \frac{1}{\sqrt{e}} 1_{\edgesE}$. The path $p$ between adjacent vertices has only one edge, so the expected norm $\lambda^2 =  1 - \frac{1}{\edgesE}$, matching the result from Corollary~\ref{cor:theta correlation}. Plugging this into Theorem~\ref{thm:R_g ring} recovers the classical result of \v{S}olc~\cite{Solc1973}: $\langle R^2_g \rangle = \frac{\edgesE^2-1}{12\edgesE}$ for the ring polymer. On the other hand, for an equilateral ring polymer of edgelength $1$, $\lambda^2 = 1$ and we get $\langle R^2_g \rangle = \frac{\edgesE+1}{12}$ as in~\cite{Zirbel2012}. Thus, we have proved a rigorous result for non-Gaussian models.

\section{Conclusion}

We have now given an explicit description of the distribution of vertex positions and displacements in a Gaussian model of topological polymers. Our description is computationally effective-- one can use it to quickly and accurately sample ensembles of polymer shapes. Further, it provides insight into the shapes of the polymers by expressing them as random linear combinations of weighted eigenvectors of the graph Laplacian of the underlying multigraph~$\graphG$. 

As mentioned in the introduction, we use this setup in~\cite{tcrw-theory,cantarellaExactFormulaContraction2025} to compute contraction factors of Gaussian topological polymers based on subdivided graphs. More generally, our approach is well-adapted to handle non-Gaussian distributions on topological polymers, and we expect it to be useful for both theoretical and computational investigation of models including self-avoidance and other more physically realistic constraints.

\acknowledgments
We are deeply grateful to the anonymous referee, whose comments inspired us to make substantial improvements to the paper. Thanks to Yasuyuki~Tezuka and Satoshi Honda for helpful discussions of topological polymer chemistry and to Fan Chung for introducing us to spectral graph theory. This paper stemmed from a long series of discussions which started at conferences at Ochanomizu University and the Tokyo Institute of Technology. Cantarella and Shonkwiler are grateful to the organizers and the Japan Science and Technology Agency for making these possible. In addition, we are grateful for the support of the Simons Foundation (\#524120 to Cantarella, \#354225 to Shonkwiler), the Japan Science and Technology Agency (CREST Grant Number JPMJCR19T4) and the Japan Society for the Promotion of Science (KAKENHI Grant Number JP17H06463).

\end{document}